\documentclass[a4paper,11pt,leqno]{amsart}

\usepackage{natbib}
\usepackage{graphicx}
\usepackage{epstopdf}
\usepackage{amsmath,amssymb}
\usepackage[hyphens]{url}
\usepackage{enumerate}

\numberwithin{equation}{section}

\newcommand{\rd}{\textrm{d}}
\newcommand{\dt}{\textrm{d}t}
\newcommand{\ds}{\textrm{d}s}

\newcommand{\dW}{\textrm{d}W_t}
\newcommand{\dWs}{\textrm{d}W_s}
\newcommand{\ex}[2][]{\mathbb{E}^{#1}\left[#2\right]}
\newcommand{\ec}[3][]{\mathbb{E}^{#1}\left[\left.#2\right|#3\right]}

\newcommand{\Ft}{\mathcal{F}_t}
\newcommand{\FT}{\mathcal{F}_T}
\newcommand{\half}{\frac{1}{2}}
\newcommand{\ind}[1]{\mathbf{1}_{\left\{ #1\right\}} }

\newcommand\bstar{\begin{eqnarray*}}
\newcommand\estar{\end{eqnarray*}}
\newcommand\be{\begin{equation}}
\newcommand\ee{\end{equation}}
\newcommand\bea{\begin{eqnarray}}
\newcommand\eea{\end{eqnarray}}

\newcommand{\vectFunc}[1]{\mathfrak{#1}}

\DeclareMathOperator*{\argmax}{arg\,max}
\DeclareMathOperator*{\argmin}{arg\,min}
\DeclareMathOperator{\diag}{diag}

 \newtheorem{proposition}{Proposition}
 \newtheorem{lemma}[proposition]{Lemma}
 \newtheorem{definition}[proposition]{Definition}
 \newtheorem{theorem}[proposition]{Theorem}

 \numberwithin{proposition}{section}

\begin{document}
\title{Risk-sensitive investment in a finite-factor model}
\author[G. Andruskiewicz, M. Davis \& S. Lleo]{Grzegorz~Andruszkiewicz, Mark~H.~A.~Davis and S\'ebastien~Lleo}
\address{Andruskiewicz, Davis: Department of Mathematics, Imperial College London SW7 2AZ, UK}
\email{gandrusz@gmail.com, mark.davis@imperial.ac.uk}
\address{Lleo: NEOMA Business School, 51100 Reims, France}
\email{lleo@neoma-bs.fr}

\date{\today}
\maketitle
\begin{abstract}
A new jump diffusion regime-switching model is introduced, which allows for linking jumps in asset prices
with regime changes. We prove the existence and uniqueness of the solution to the risk-sensitive asset management
criterion maximisation problem in this setting. We provide an ODE for the optimal value function, 
which may be efficiently solved numerically. Relevant probability measure changes are discussed in the appendix.
The recently introduced approach of Klebaner and Liptser (2013) is used to prove the martingale property of the relevant density processes.
\end{abstract}

\section{Introduction}\label{sect:intro}
Recently, especially after the the latest credit crisis,
many hedge funds and portfolio managers have been taking interest in improved
modelling of asset returns in the market. They use the distribution of returns
to derive trading strategies that optimise the trade-off between
high expected returns and volatility.
Fund managers always have a view of the market, but need  tools to optimise
their portfolios. They constantly follow market conditions and regularly recalibrate  model parameters and revise asset allocations. 

In this paper we introduce a new regime-switching factor model for asset prices. There are a finite number of regimes and within each regime the asset processes follow jump diffusions. Regime switching follows an autonomous continuous-time Markov chain. Many other authors, starting with \cite{merton1976} and including \cite{br2004} and \cite{sc2009}, have used Markov chain factor models, but a critical new feature here is that a change of regime may coincide with a jump in asset prices. This is important because it forces the investor to hedge against possible regime changes. In our model, as in most others, there is frictionless trading, so if asset prices never jump at a regime shift the investor can simply wait for the shift to happen and then rebalance instantaneously, whereas if the shift is accompanied by price jumps then some defensive measures must be taken in advance. As an example, it is clearly unreasonable to suppose that it would have been possible to rebalance a portfolio between the recent announcement that CHF will cease to be pegged against EUR and the consequent jump in the foreign exchange market.

After defining the market model we optimise a risk-sensitive criterion in the given setting.  For optimal investment over a finite time horizon $[0,T]$ the risk-sensitive criterion is
\be J_\theta (v_0,h,T)=-\frac{1}{\theta}\log\mathbb{E}\left[ e^{-\theta \log V(v_0,h,T)}\right]\label{rsc}\ee
where $v_0$ is the initial capital, $V(v_0,h,T)$ is the portfolio value at time $T$ resulting from an investment strategy $h$ and $\theta$ is a risk-aversion parameter; we always take $\theta\geq 0$. In all standard models, including the ones here, $V(v_0,h,T)=v_0V(1,h,T)$, so
\[ J_\theta (v_0,h,T)=\log v_0-\frac{1}{\theta}\log\mathbb{E}\left[ e^{-\theta \log V(1,h,T)}\right].\]
This shows that optimisation does not depend on $v_0$, and we can and will harmlessly normalise to $v_0=1$. Taking a formal Taylor expansion of $J_\theta$ around $\theta=0$ gives
\[ J_\theta (v_0,h,T)=\mathbb{E}[\log V(T)]-\frac{\theta}{2}\mathrm{var}[\log V(T)]+O(\theta^2).\]
If we define the realised growth rate, or return, $R_T$ by $V(T)=e^{R_TT}$ this becomes
\be\frac{1}{T}J_\theta (v_0,h,T)=\mathbb{E}[R(T)]-\frac{\theta T}{2}\mathrm{var}[R_T]+O(\theta^2T^2).\label{lagrange}\ee
Ignoring for a moment the remainder term, we can regard $\theta T/2$ as a Lagrange multiplier for the optimisation problem of maximising the expected return subject to a constraint on the variance of the return---a continuous-time
equivalent of mean-variance analysis introduced by \cite{mar52}, which is still the dominant technique in the market.
 A further interpretation of \eqref{rsc} is to note that $\exp(-\theta\log V)=V^{-\theta}$, so maximising $J_\theta(1,h,T)$ is equivalent to minimising $\mathbb{E}[V(T)^{-\theta}]$, i.e. maximising expected utility for the power utility function $x^{\gamma}/\gamma$ with $\gamma=-\theta$.

 Application of risk-sensitive control to asset management problems was pioneered by  \cite{bielecki2003} and a considerable literature has developed since then. References can be found in \cite{mdsl2014}, and one of the models studied there and in the paper \cite{davis2013} includes jumps in both
the assets and the factors; however, for technical reason we were obliged to exclude simultaneous asset and factor jumps, so the model presented here is not a special case.

The closest results to ours are in  \cite{frey2012}, where a standard regime-switching model is assumed.
But in their model there are no jumps in the asset prices, instead the authors assume that 
the regimes are not observable and include stochastic filtering in their analysis. 
In contrast, we assume that the investor knows which state the system is currently in---in fact we believe this 
information is part of the investor's view of the market---together with all the other parameter values of the model.
The question of how to determine these values is outside of the scope of this article---and indeed this
is the very skill that allows the investors to generate \emph{alpha}. Some further remarks will be found in Section \ref{sect:conc} below.
\cite{bielecki1999a} consider a Markov chain model, which captures explicitly the correlation between 
the state of the economy and the returns of the assets. Their technical approach is different to the one presented in this paper, being formulated in a discrete time setting.

The paper is laid out as follows. In Section \ref{sect:market} we formulate our basic model for factors and asset prices, 
and in Section \ref{sect:criterion} we discuss the risk-sensitive criterion and the set of admissible investment strategies. 
The main results of the paper are contained in Section \ref{sect:trick} where we study the HJB equation. 
Because of the finite-state factor specification, this is an \emph{ordinary differential equation}, 
as opposed to the partial integro-differential equations that arise in jump-diffusion factor models such as those in \cite{davis2013}, 
and this is the major advantage of our approach. We establish existence and uniqueness in Theorem \ref{theorem:control}. 
 In Section \ref{sect:mft} we consider the relationship between risk-sensitive optimal strategies and the Kelly `growth-optimal' strategy.  Further remarks on application of the model will be found in the concluding section, Section \ref{sect:conc}.

As will be seen, our whole approach is based on measure changes and the Dol\'eans-Dade (generalized Girsanov) exponential martingale. We cover the required information in two Appendices, 
the second of which is, we believe, the first application in an applied context of a new approach to showing that the stochastic exponential has expectation 1, due to  \cite{klebaner2011}.

\section{Market}\label{sect:market}
Our model is constructed on a probability space $(\Omega,\mathcal{G},\mathbb{P})$ carrying the following three independent objects, which are specified in detail below.

(i) An $N$-state continuous-time homogeneous Markov chain $x_t, t\in[0,T]$;

(ii) A Brownian motion $W_t, t\in[0,T]$ in $\mathbb{R}^m$;

(iii) A sequence $U_1,U_2,\ldots$ of i.i.d. random variables, where $U_1$ is uniformly distributed on $[0,1]^m$.

\noindent We propose a factor-based model for asset prices on a fixed time interval $[0, T]$. The factor process is the Markov chain $\{x_t\}$, with states in $\mathcal{N}=\{1, \dots, N\}$ 
and generator $Q$ in the real-world probability measure $\mathbb{P}$. 
It is convenient to identify the process $x_t$
with a  process $X_t\in \mathbb{R}^N$ where $X_t=e_k$, the $k$-th unit coordinate vector, when $x_t=k$.
Note that the jumps of $X_t$ arrive according to the state-dependent Poisson process $\{\Lambda_t\}$ with intensity $\lambda(X_t)$ at
time $t$, where from the theory of Markov chains the jump intensity is defined by the generator matrix: $\lambda(i)=-Q_{ii}$ for every $i$.

First let $\mathfrak{F}$ be the class of density functions $f:\mathbb{R}^m\times\mathcal{N}\times\mathcal{N}\rightarrow \mathbb{R}$  of asset jump sizes
that satisfy the following conditions:
\begin{enumerate}[(i)]
 \item $f(\cdot, i, j)$ is a density function for every $i, j \in \mathcal{N}, i\neq j$.
 \item $f(z, i, j)=0$ for every $i, j \in \mathcal{N}, i\neq j$ and $z\notin \mathcal{Z}^i \subseteq [z^i_{min}, z^i_{max}]^m$, where $z^i_{min}>-1$ and $z^i_{max}<\infty$.
 \item $\sum_{i, j; i\neq j} \int_{\mathbb{R}^m} |z| f(z; i, j) \rd z <\infty$
\end{enumerate}
Note that if $z^i_{min}, z^i_{max} <0$ then in the state $i$ we allow only downward jumps of asset prices. Analogously,
$z^i_{min}, z^i_{max} >0$ means that only upward jumps are allowed in this state.
The last condition guarantees that jumps near the boundary of $\mathcal{Z}^i$ can indeed happen with positive probability, and is needed in the proof of Proposition \ref{prop:unique_h} below.
Define a process family $M_t \in \mathbb{R}^m$ to be:
\begin{equation}\label{eq:Mt}
M_t = \sum_{T_{i}\leq t} Z_i - \int_0^t \sum_{j\neq X_{s-}}\int_{\mathbb{R}^m} z f(z; X_{s-}, j) Q(X_{s-}, j) \rd z \ds
\end{equation}
where the jump times $T_i$ coincide with jumps in process $X$ and random variables $Z_i$ are conditionally independent of $\{\Lambda\}$ and each other, 
and have an $m$-dimensional distribution with density $f(\cdot ; X_{t-}, X_{t})\in \mathfrak{F}$, 
depending on the state before and after the jump. $Z_i$ can be constructed from the uniform random variable $U_i$ by the usual inverse mapping procedure in each dimension involving the density function $f(\cdot ; X_{t-}, X_{t})$.
Note that $\{M\}$ is a martingale family in the filtration $\Ft^{M, X}=\sigma(\{M_s\}_{0\leq s \leq t}, \{X_s\}_{0\leq s \leq t})$, generated by both $M$ and $X$.
Denote the expected value of $Z$ if $X$ jumps from state $i$ to $j$ as:
\begin{equation}\label{eq:xi}
\xi(i, j) =  \int_{\mathbb{R}^m} z f(z; i, j) \rd z
\end{equation}
and define centred jumps as:
\begin{equation}
Y_i = Z_i - \xi(X_{T_{i}-}, X_{T_{i}}).
\end{equation}
The quantities above are well defined because of the integrability assumptions in the definition of $\mathfrak{F}$. Then $M_t$ may be written as:
\begin{equation}\label{eq:M}
M_t = \sum_{T_{i}<t} Y_i + \sum_{T_{i}<t} \xi(X_{T_{i}-}, X_{T_{i}}) - \int_0^t \sum_{j\neq X_{s-}}\xi(X_{s-}, j) Q(X_{s-}, j) \ds. 
\end{equation}
Note that $M_t$ is a Piecewise Deterministic Process (PDP), see \cite{davis1993} for an in-depth discussion.

Moving to the asset model, there are $m$ risky assets in the market, given by:
\begin{equation}
 \frac{\rd S_t^i}{S_{t-}^i} = \mu_i(t, X_t) \dt + \Sigma_i(t, X_t)\dW + \rd M_t^i,\quad S_0^i=s_i,
\end{equation}
for $i=1,\ldots,m$ and initial prices $s_i>0$, where $M_t^i$ is the $i$-th coordinate of $M_t$.
We assume that for some $\epsilon>0$, $\Sigma(t, i)  \Sigma(t, i)' > \epsilon I$ for every $t$, $i$. 
Moreover,  $\mu(t, i)$ and $\Sigma(t, i)$ are continuous functions of $t$ for every $i\in \mathcal{N}$.
The solution to the SDE above is given by:
\begin{equation}
\begin{split}
  S^i_t &= \exp\left( \int_0^t \mu_i(s, X_s) \ds - \half\int_0^t \Sigma_i(s, X_s)\Sigma_i(s, X_s)' \ds + \int_0^t \Sigma_i(s, X_s) \dWs \right)\\
  &\times \exp\left(- \int_0^t \sum_{j\neq X_{s-}}\xi_i(X_{s-}, j) Q(X_{s-}, j) \ds \right)\prod_{0\leq s\leq t} (1+Z^i_s).
\end{split}
\end{equation}
The stock prices are guaranteed positive, thanks to the assumption that $Z^i\geq z_{min} >-1$ in the definition of $\mathfrak{F}$. The upper bound, $Z^i\leq z_{max}<\infty$,
is introduced to allow short-selling of the stocks without a possibility of jumping to bankruptcy. 
Note that some authors work with jumps $\zeta$ defined by $\zeta=\log(1+Z)$
instead, and as a result $\zeta$ can take any real value; for example in the \cite{merton1976} jump diffusion model the $\zeta$ are Gaussian, and in \cite{kou2002} they are doubly-exponential.
The risk free asset is assumed to grow at a rate dependent on the factor process as well\footnote{Note that this constitutes a very simple model for stochastic interest rates.}:
\begin{equation}
 \frac{\rd S^0_t}{S^0_t} = r(t, X_t)\dt, \quad \quad S^0_0=1.
\end{equation}
Let $\Ft^S = \sigma(\{S_u\}_{0\leq u \leq t})$ denote the natural filtration generated by the asset processes and let $\Ft^X = \sigma(\{X_u\}_{0\leq u \leq t})$ be the filtration generated by the factor process $X$. 
As mentioned in the introduction, in this paper we work in the filtration generated by both assets and the factor process: 
\begin{equation}\label{eq:Ft}
\Ft = \sigma(\{S_u, X_u\}_{0\leq u \leq t}) 
\end{equation}
Because the jumps in the assets correspond to the jumps in the martingale $M$, the filtration generated jointly by $M$ and $X$ is a subset of the full filtration: $\Ft^{M, X} \subseteq \Ft$.
$\{M\}$ is also a martingale in $\Ft$.

Note that in our model the state variable $X_t$ not only tracks the current market regime, but also drives the jumps in the asset prices. In practice, the latter jumps
are expected to happen much more often than regime changes.
This could be easily modelled by a two dimensional state process $X$, see Figure \ref{fig:exampleRegimes} for an example.
Because the number of states is assumed to be finite, multidimensional Markov chain may be mapped to a single-dimensional chain,
and hence this scenario is handled by our model out of the box. 

\begin{figure}[htb]
\includegraphics[scale=0.7]{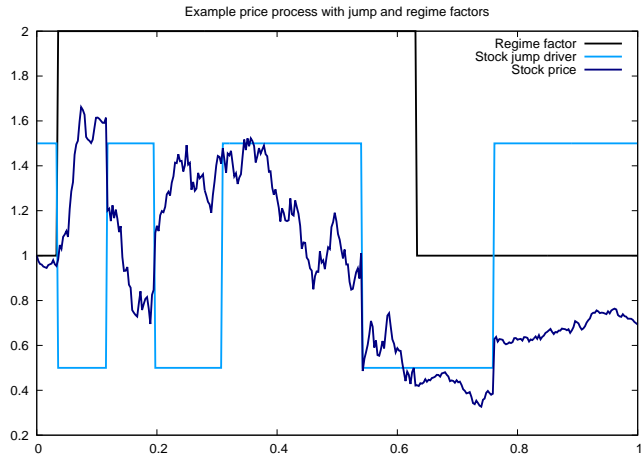}
\caption{\label{fig:exampleRegimes}Example of a path of a stock price process together with corresponding market regime and jump driver paths. }
\end{figure}

\section{The optimal investment problem}\label{sect:criterion}
At every point in time $t$ the investor chooses an asset allocation $h_t$,  an $m$-vector whose $j$-th component 
denotes the proportion  of the portfolio value invested in the  $j$-th asset, $j=1, \dots, m$. Because the risk-sensitive criterion
is defined as $-\infty$ for negative portfolio values, we need to add assumptions on $h$
so that the portfolio never jumps into negative territory. To this end we define for each $i\in\mathcal{N}$ 
\begin{equation}\label{eq:admissibility}
 \mathcal{J}(i) = \left\{h \in \mathbb{R}^m: h'\psi>-1 \quad  \forall \psi\in \mathcal{Z}^i \right\},
\end{equation}
where $\mathcal{Z}^i \subseteq [z^i_{min}, z^i_{max}]^m$ is defined in Section \ref{sect:market}.
Note that $\mathcal{J}(i)$ is non-empty and bounded, because $\mathcal{Z}^i$ is bounded and greater than $-1$ for any $i$.
\begin{definition}\label{def1}
$\mathcal{H}_0$ is the class of functions $h:[0,T]\times\Omega\rightarrow \mathbb{R}^m$ such that
\begin{enumerate}[(i)]
 \item $h_t\in\mathcal{J}(X_{t-})$ for every time $t\in[0, T]$
 \item $h(t,\omega)$ is a predictable process with respect to the filtration $\Ft$ defined in (\ref{eq:Ft}).
\end{enumerate}
\noindent $\mathcal{M}_0$ is the set of measurable function $\hat{h}:[0,T]\times\mathcal{N}\to\mathbb{R}^m$ such that $\forall (t,i), \,\hat{h}(t,i)\in\mathcal{J}(i)$.
Each  $\hat{h}\in\mathcal{M}_0$ defines an element of $\mathcal{H}_0$ by $h(t,\omega)=\hat{h}(t, X_{t-}(\omega))$. 
\end{definition}
The trading portfolio is assumed to be self-financing, and hence
all the changes in portfolio value are caused by the changes in the underlying asset prices:
\begin{equation}
 \rd V_t = \left(\frac{V_t h_t}{S_t}\right)'\rd S_t +\frac{V_t(1-h'\mathbf{1})}{S^0_t}\rd S^0_t,
\end{equation}
where $\frac{V_T h_t}{S_t}$ is the vector containing the number of units in each asset and the division is interpreted componentwise.
$1-h'\mathbf{1}$ is the proportion invested (or borrowed) in the money market account.
After substituting in the previous formula we get:
\begin{equation}
 \frac{\rd V_t}{V_{t-}} = r(t, X_t) \dt + h'_t(\mu(t, X_t)-r(t,X_t)\mathbf{1})\dt + h'_t\Sigma(t,X_t)\dW + h'_t \rd M_t, \quad V_0 = 1.
\end{equation}
Note that the process $h_t'\Sigma\dW + h_t' dM_t$ is a local martingale, hence the solution of the SDE above is the stochastic exponential
of the martingale part with the drift:
\begin{equation}\label{V}
\begin{split}
  V_t &= \exp\left( \int_0^t r_s+h'(\mu_s - r_s\mathbf{1}) \ds - \half\int_0^t h'_s\Sigma\Sigma'h_s \ds + \int_0^t h'_s\Sigma \dWs + \int_0^t h_s' dM_s\right)\\
  &\times \prod_{0\leq s\leq t} (1+h_s'Z_s)e^{-h_s'Z_s},
\end{split}
\end{equation}
where $r_s:=r(s, X_s)$ etc.; see \citet[pp.84-85]{protter2005} for detailed calculations.
The logarithm of the value process is thus given by:
\bea
\log V_t &=& \int_0^t r_s+h'(\mu_s-r_s\mathbf{1}) \ds - \half\int_0^t h'_s\Sigma\Sigma'h_s \ds + \int_0^t h'_s\Sigma \dWs + \int_0^t h_s' dM_s\nonumber \\
&&\quad + \sum_{0\leq s\leq t} \left\{ \log (1+h_s'Z_s) - h_s'Z_s \right\} \label{eq:logVt}\\
&=&\int_0^t [r_s+h_s'(\mu_s-r_s\mathbf{1})] \ds - \half\int_0^t h'_s\Sigma\Sigma'h_s \ds + \int_0^t h'_s\Sigma \dWs\nonumber \\
&&\quad + \sum_{0\leq s\leq t} \log (1+h_s'Z_s) - \int_0^t \sum_{j\neq X_{s-}} h_s'\xi(X_{s-}, j)Q(X_{s-},j) \ds.\label{eq:logVt2}
\eea

\section{Optimisation: the HJB equation}\label{sect:trick}
The investor maximises the risk-sensitive criterion (\ref{rsc}),
using strategies that ensure $V>0$ at all times. This condition is holds for any strategy in the class $\mathcal{H}_0$
because for every realisation of $Z$ the conditions guarantee that
\begin{equation}\label{eq:assPositiveV}
\forall_t  \quad h_t'Z_t>-1\quad\quad\textrm{a.s.}
\end{equation}
Following the idea from \cite{kuroda2002}, we can write the term under the expectation in the criterion (\ref{rsc}) as:
\begin{equation}
 e^{-\theta \log V_T} =  \exp \left(\theta \int_0^T g(t, X_{t-}, h_t) \dt \right)\chi^h_T,
\end{equation}
where
\begin{equation}\label{eq:chi_t}
\begin{split}
 \chi^h_t &= \exp \left( -\theta \int_0^t h_s'\Sigma\dWs -\frac{\theta^2}{2} \int_0^t h_s'\Sigma \Sigma'h_s \ds \right)\prod_{0<s\leq T} (1+h_s'Z_s)^{-\theta}\\
&\times \exp \left(- \int_0^T \sum_{j\neq X_{s-}} \int_{\mathcal{Z}}[(1+h_s'z)^{-\theta} -1]f(z; X_{s-}, j) \rd z Q(X_{s-},j) \ds \right)
\end{split}
\end{equation}
and, for each $i\in\mathcal{N}$, $g(\cdot,i,\cdot):[0,T]\times\mathcal{J}(i)$ is given by
\begin{equation}\label{eq:g}
\begin{split}
 g(t, i, h) &= \half (\theta +1)h'\Sigma(t,i)\Sigma(t,i)'h - r(t, i) - h'(\mu(t,i)-r(t,i)\mathbf{1})\\
 &+\sum_{j\neq i} Q(i, j) \left[ \frac{1}{\theta}\int_{\mathcal{Z}}[(1+h'z)^{-\theta} -1]f(z; i, j) \rd z + h'\xi(i, j)\right]. 
\end{split}
\end{equation}
The process $\chi^h_t$ is an exponential local martingale  \citep[Theorem II.37]{protter2005}.

\begin{definition}\label{def:admissible}
 A trading strategy $h:\Omega \times \mathbb{R} \rightarrow \mathbb{R}^m$ is \emph{admissible}  if 
it is in class $\mathcal{H}_0$ of Definition \ref{def1} and the following condition holds:
\begin{equation}\label{eq:integrability}
 \int_{\mathcal{Z}}(1+h'z)^{-\theta} \sum_{j\neq X_{t-}}Q(X_{t-}, j) f(z; X_{t-}, j) \rd z <\infty.
\end{equation}
The set of admissible strategies will be denoted by $\mathcal{H}$.
\end{definition}

\begin{definition}\label{def:markovStrategies}
The trading strategy $h:[0,T]\times\Omega  \rightarrow \mathbb{R}^m$ is a \emph{Markov strategy} if 
it is in class $\mathcal{H}$ defined above and
\begin{equation}\label{eq:markovControl}
h(t,\omega) = \tilde{h}(t, X_{t-}(\omega)),
\end{equation}
for some $\tilde{h}\in\mathcal{M}_0$. The set of Markov strategies will be denoted by $\mathcal{M}$.
\end{definition}

\begin{proposition}\label{prop:martingale}
For any fixed trading strategy $h\in\mathcal{H}$, the stochastic process $\chi^h$ is a martingale with $\ex{\chi^h_T}=1$.
\end{proposition}
\begin{proof}
See Appendix \ref{sect:proof:prop:martingale}.
\end{proof}
Using the proposition above and measure change theory summarised in Appendix \ref{sect:measureChange},
we use the martingale $\{\chi^h_t\}$ to change the probability measure:
\begin{equation}\label{eq:control:measureChange}
 \left.\frac{\rd \mathbb{P}^h}{\rd \mathbb{P}}\right|_{\FT} = \chi^h_T.
\end{equation}
Under the new measure $\mathbb{P}^h$ the risk-sensitive criterion becomes:
\begin{equation}
\begin{split}
J_\theta(h,T) &= -\frac{1}{\theta}\log\ex{e^{-\theta \log V_T}}\\
&=-\frac{1}{\theta}\log\ex[h]{e^{\theta \int_0^T g(t, X_t, h_t) \dt}}
\end{split}
\end{equation}
and $Q^h = [Q^h(i, j)]$ becomes a generalised generator of $X$ with elements:
\begin{equation}\label{eq:qh}
\begin{split}
Q^h(i, j)(t) &= Q(i, j)\left[\int_{\mathcal{Z}^i}[(1+h'z)^{-\theta}]f(z; i, j) \rd z\right] ,\\
Q^h(i, i)(t) &= -\sum_{j \neq i} Q^h(i, j)(t).
\end{split}
\end{equation}
We can write the optimal value function as:
\begin{equation}\label{eq:value_f_v}
\begin{split}
v(t, i) &= \sup_{h\in\mathcal{H}} -\frac{1}{\theta}\log \mathbb{E}_{t, i}^h\left[ e^{\theta \int_t^T g(s, X_s, h_s) \ds}\right] \\
&= -\frac{1}{\theta} \log u(t, i),
\end{split}
\end{equation}
where
\begin{equation}\label{eq:value_f_u}
u(t, i) = \inf_{h\in\mathcal{H}} \mathbb{E}_{t, i}^h\left[ e^{\theta \int_t^T g(s, X_s, h_s) \ds}\right]
\end{equation}
and where $\mathbb{E}_{t, i}^\cdot\left[ \cdot \right]$ denotes the conditional expectation $\mathbb{E}^\cdot\left[ \cdot | X_t = i \right]$.
Note that, thanks to the measure change and the normalisation of the initial investment,
the value function doesn't depend on the value of the portfolio at time $t$, it only depends 
on the state of the factor process $X$.
In the remainder of the paper we solve for the function $u$, and the original value function $v$ can be easily obtained using the formula above.

Below we include a few propositions that make it easier to understand what $u$ is and what its basic properties are.
\begin{proposition} \label{prop:control:u_representation}
The function $u$ may also be expressed as:
\begin{equation}
u(t, i) = \inf_{h\in\mathcal{H}}  \mathbb{E}_{t, i}\left[e^{-\theta \log (V_T/V_t)}\right] = \inf_{h\in\mathcal{H}} \mathbb{E}_{0, i}\left[ (V^t_{T-t})^{-\theta}\right].
\end{equation}
In the last expression, $V^t_s$ is the portfolio value process as in \eqref{V} but with time-shifted coefficients $\mu^t(s,i):=\mu(t+s,i)$ etc., for $s\in[0,T-t]$.
\end{proposition}
\begin{proof}
By Remark 5.2 from \cite{bouchard2011} we can restrict ourselves to controls $h$ that are independent of $\Ft$. To prove the first equality
we apply the measure change defined by (\ref{eq:control:measureChange}) to equation (\ref{eq:value_f_u}) backwards:
\begin{equation}
\begin{split}
\mathbb{E}_{t, i}^h\left[ e^{\theta \int_t^T g(s, X_s, h_s) \ds}\right] &=  \mathbb{E}_{t, i}^h\left[ e^{\theta \int_t^T g(s, X_s, h_s) \ds} \chi_T^h / \chi_t^h\right]\\
&=\mathbb{E}_{t, i}\left[e^{-\theta \log (V_T/V_t)} \right].
\end{split}
\end{equation}
Note that $V_T/V_t$, conditionally on $X_t=i$, is independent of $\Ft$. 
The second equality follows from the Markov property of the system and simple algebraic transformations.
\end{proof}
\begin{proposition}\label{prop:missingIngredient}
The function $u(t, i)$ is increasing in time parameter $t$.
\end{proposition}
\begin{proof}
From Proposition \ref{prop:control:u_representation}, $u(t, i)$ is the value function corresponding to the optimal investment over the time horizon
$[0, \tau]$, where $\tau = T - t$. Let $0<\delta \leq t$. We can extend the trading strategy from $[0, \tau]$ to $[0, \tau+\delta]$
by putting all the money in the bank at $\tau$, giving $V_{\tau+\delta} \geq V_\tau e^{r\delta}$, where $r$ is the minimal interest rate.
Hence the expectation is decreasing in the time horizon:
\begin{equation}
\ex{V_{\tau+\delta}^{-\theta}} \leq \ex{(V_\tau e^{r\delta})^{-\theta}} = \ex{V_{\tau}^{-\theta}} e^{-\theta r\delta}\leq \ex{V_{\tau}^{-\theta}},
\end{equation}
which implies that the function $u$ is increasing in the real time variable $t$:
\begin{equation}
u(t-\delta, i) \leq u(t, i).
\end{equation}
\end{proof}

\begin{proposition}\label{prop:domain}
The range of $u$ is a compact set $\mathcal{U} \subseteq [u_{min}, u_{max}]^N$, such that $0 < u_{min} \leq u_{max} < \infty$, 
where $u$ is defined in (\ref{eq:value_f_u}).
\end{proposition}
\begin{proof}
Because the value function is defined in terms of minimisation of the expectation $\ex{e^{-\theta \log V_T}}$ 
we can bound the value from below and above. First note that the function $g$ defined in (\ref{eq:g}) 
is bounded from below for any $t, i, h$, and let $g_{min} = \inf \{g(t, i, h) : t\in [0, T], i \in \mathcal N, h \in \mathcal{J}_h \}$. Then 
\begin{equation}
u(t, i) \geq e^{\theta g_{min} (T-t)} >0.
\end{equation}
The argument showing the upper bound is a bit more subtle. The expectation $\ex{e^{-\theta \log V_T}}$ is large
when the portfolio value on average performs badly. Note, however, that the investor always has the option
to put all his wealth in the money market account. This pays a guaranteed return, which is different in every regime,
but it is at least $r_{min} = \inf \{ r(i, t) : i \in \mathcal N, t\in [0, T] \}$. Thanks to the minimisation operator in the definition of $u$, 
the upper bound is given by:
\begin{equation}
u(t, i) = \inf_{h\in\mathcal{H}}  \ex{e^{-\theta \log (V_T/V_t)}} \leq e^{-\theta r_{min}(T-t)} 
\end{equation}
Note that $g_{min} \leq \inf \{g(t, i, 0) : t\in [0, T], i \in \mathcal N \} = -r_{min}$, hence $\mathcal{U}$ is not empty.
The solution is defined on the finite time interval $[0, T]$, which finishes the argument.
\end{proof}

\subsection{HJB equation: general case}\label{sec_general}
In this section $\vectFunc{u}$ will denote an $N$-vector valued function on $[0,T]$ and we use $\vec{u}$ for a generic $N$-vector.
We need to solve the HJB equation, which in this case is the ODE in $\mathbb{R}^N$:
\begin{equation}\label{eq:ode_min}
\frac{\rd \vectFunc{u}}{\rd t} + \inf_{h\in\mathcal{H}} \left\{ A(\vectFunc{u}(t), h) \right\} =0,
\end{equation}
with
\begin{equation}\label{eq:Auh}
 A(\vec{u}, h) = Q^h \vec{u} + \theta \diag (g)\vec{u},
\end{equation}
where $Q^h(i)$ is the $i$th row of $Q^h$. 
All the functions of $X$ are interpreted as corresponding vectors and $\theta>0$. The $\inf$ operator is
interpreted componentwise, that is the vector $h_i$ minimises the $i$th element of $A(\vectFunc{u}(t), h)$.
The boundary condition is given by:
\begin{equation}\label{eq:boundaryCondition}
\vectFunc{u}_i(T) = 1, \quad i\in\mathcal{N}.
\end{equation}
Our objective is to show that $\vectFunc{u}_i(t)=u(t,i)$, the value function of \eqref{eq:value_f_u}. The following is the main result of this paper.
\begin{theorem}\label{theorem:control}
Suppose the market is defined as in Section \ref{sect:market}, admissible strategies are as in Definition \ref{def:admissible} and $\theta>0$. 
Then the HJB equation (\ref{eq:ode_min}) with final condition (\ref{eq:boundaryCondition}) defined above,
 has a unique solution on $[0, T]$, which coincides with the value function defined in (\ref{eq:value_f_u}). The Markov control $\tilde{h}(t, X_{t-})=h^*(\vectFunc{u},t,X_{t-})$ is optimal in the class of admissible controls $\mathcal{H}$.
\end{theorem}
\begin{proof}
The theorem follows from Propositions \ref{prop:verification} and \ref{prop:solution} below. 
\end{proof}
Following the standard approach, we first find the optimal strategy for every time $t$ and state $i$, where we take
the value function $u$ as an argument:
\begin{equation}\label{eq:h_star}
 h^*(\mathfrak{u}, t, i) = \argmin_{h\in\mathcal{J}(i)} \left\{ A(\vectFunc{u}(t),h)(t, i)\right\},
\end{equation}
By substituting (\ref{eq:g}) and (\ref{eq:qh}) into (\ref{eq:Auh}) we get an explicit formula for the operator $A$:
\begin{equation}\label{eq:A}
\begin{split}
A(\vec{u},h)(t, i) &=Q^h(i) \vec{u} +\theta g(t, h, i) \vec{u}(i)\\
&= \sum_{j\neq i}Q(i, j) \vec{u}(j) \int_{\mathcal{Z}}(1+h'z)^{-\theta}f(z; i, j) \rd z  \\
&+\half \vec{u}(i) \theta (\theta +1)h'\Sigma(t,i)\Sigma(t,i)'h \\
&-\theta \vec{u}(i)h'[\mu + \sum_{j\neq i}Q(i, j)\xi(i, j) - r\mathbf{1}] \\
& -\theta \vec{u}(i)r(t, i) -\vec{u}(i)\sum_{j\neq i}Q(i, j),
\end{split}
\end{equation}
where $h$ is a function of time $t$ and state $i$.
This operator is linear in $\vec{u}$, and can be explicitly written as matrix multiplication $A(h)\vec{u}$, with:
\begin{equation}\label{eq:matrixA}
\begin{split}
A(h)_{ij} &= Q(i, j) \int_{\mathcal{Z}^i}(1+h'z)^{-\theta}f(z; i, j) \rd z \quad\quad \textrm{for $i \neq j$}\\
A(h)_{ii} &= \half\theta (\theta +1)h'(t,i)\Sigma(t,i)\Sigma(t,i)'h(t,i) \\
&- \theta h'(t, i)(\mu+ \sum_{j\neq i}Q(i, j)\xi(i, j, t)-r\mathbf{1})\\
&-\theta r(t, i) -\sum_{j\neq i}Q(i, j).
\end{split}
\end{equation}
Let:
\begin{equation}\label{eq:caligraphicA}
\mathcal{A}(\vec{u})(t, i) = \inf_{h\in\mathcal{J}_h} \left\{ A(\vec{u},h)(t, i) \right\}.
\end{equation}
\begin{proposition}\label{prop:unique_h}
The operator $A(h)\vec{u}$, as a function of $h$, for all $\vec{u}\in\mathcal{U}_0$, has a unique minimum in $\mathcal{J}_h$. 
Hence the operator $\mathcal{A}(\vec{u})$ is well defined. Moreover, the optimal strategy $\tilde{h}$ satisfies the condition (\ref{eq:integrability}).
\end{proposition}
\begin{proof}
First note that the first term in (\ref{eq:A}),
\begin{equation}
\sum_{j\neq i}Q(i, j) \vec{u}(j) \int_{\mathcal{Z}^i}(1+h'z)^{-\theta}f(z; i, j) \rd z,
\end{equation}
is positive and convex and hence, 
for every index $i$ and time $t$, $A(\vec{u},h)(t, i)$ is a finite sum of continuous functions, convex in $h$ and bounded from below, 
hence it is also continuous, convex in $h$ and bounded from below.
Note that the this term can be written as:
\begin{equation}
 \int_{\mathcal{Z}^i}(1+h'z)^{-\theta} \sum_{j\neq i}Q(i, j) \vec{u}(j) f(z; i, j) \rd z,
\end{equation}
 There are two cases, depending on whether
\begin{equation}
\lim_{h\rightarrow \partial\mathcal{J}(i)} \int_{\mathcal{Z}^i}(1+h'z)^{-\theta} \sum_{j\neq i}Q(i, j) f(z; i, j) \rd z
\end{equation}
is finite or equal to $+\infty$. In the first case a minimizing element exists, which
might or might not be on the boundary. In the second case, some analysis
shows that the derivative of the operator A diverges to $+\infty$ as $h$ approaches
the boundary, and hence the minimum is an interior point. In both cases  the condition (\ref{eq:integrability}) is satisfied.
\end{proof}

The HJB equation (\ref{eq:ode_min}) may be written as:
\begin{equation} \label{eq:ode}
\frac{\rd \vectFunc{u}}{\rd t} + \mathcal{A}(\vectFunc{u}(t)) = 0
\end{equation}
with final condition:
\begin{equation}\label{eq:finalCondition}
 \vectFunc{u}(T) = \mathbf{1}.
\end{equation}
\begin{proposition}[Verification theorem]\label{prop:verification}
 If $\tilde{u}$ is a solution to ODE (\ref{eq:ode}) with the final condition (\ref{eq:finalCondition})
 on some interval $[t, T]$, and $\tilde{h}(t) = \tilde{h}(t, X_{t-} \in \mathcal{M}$ is the corresponding Markov control,
 then $\tilde{u}$ is the value function (\ref{eq:value_f_u})  
 and $\tilde{h}$ is optimal in the class of all admissible trading strategies $\mathcal{H}$.
\end{proposition}
\begin{proof}
Using Ito's lemma, we get the following relationship:
\begin{equation}
\begin{split}
\rd \left(e^{\theta\int_0^t g(s, X_s, \tilde{h}_s)\ds} \tilde{u}(t, X_t)\right) &= e^{\theta\int_0^t g(s, X_s, \tilde{h}_s)\ds} \tilde{u}(t, X_t)\theta g(t, X_t, \tilde{h}_t) \dt\\
&\quad +e^{\theta\int_0^t g(s, X_s, \tilde{h}_s)\ds} \left[ \frac{\partial \tilde{u}}{\partial t}(t, X_t)\dt + \Delta \tilde{u}(t, X_t) \right]\\
&=e^{\theta\int_0^t g(s, X_s, \tilde{h}_s)\ds} \tilde{u}(t, X_t)\theta g(t, X_t, \tilde{h}_t) \dt\\
&\quad +e^{\theta\int_0^t g(s, X_s, \tilde{h}_s)\ds} \left[ \frac{\partial \tilde{u}}{\partial t}(t, X_t) + Q^{\tilde{h}}\tilde{u}(t, X_t)\right] \dt\\
&\quad + e^{\theta\int_0^t g(s, X_s, \tilde{h}_s)\ds} \left[\Delta \tilde{u}(t, X_t)-Q^{\tilde{h}}\tilde{u}(t, X_t)\dt \right],
\end{split}
\end{equation}
where $Q^{\tilde{h}}$ is the generator of the process $\{X\}$ in the measure corresponding to the trading strategy $\tilde{h}$.
After integrating over $[t, T]$, multiplying both sides by $e^{-\theta\int_0^t g(s, X_s, \tilde{h}_s)}$ and rearranging, we get:
\begin{equation}
\begin{split}
\tilde{u}(t, X_t) &=  e^{\theta\int_t^T g(s, X_s, \tilde{h}_s)\ds}\tilde{u}(T, X_T) \\
 &-\int_t^T e^{\theta\int_t^s g(u, X_u, \tilde{h}_u) \rd u}\left[ \frac{\partial \tilde{u}}{\partial s}(s, X_s) + A_i(\tilde{h}_s, \tilde{u}_s)\right]\ds\\
 &-\left[\sum_{t\leq s \leq T} \Delta \tilde{f}(s, X_s) - \int_t^T Q^{\tilde{h}}\tilde{f}(s, X_s)\ds \right],
\end{split}
\end{equation}
where $\tilde{f}(t, X_t) = e^{\theta\int_t^s g(u, X_u, \tilde{h}_u) \rd u} \tilde{u}(t, X_t)$, 
the operator $A$ is defined in (\ref{eq:Auh}) and $A_i(\cdot, \cdot)$ denotes the $i$th element of the resulting vector.
Note that the term on the last line in the equation above is a martingale.
By taking a conditional expectation $\mathbb{E}_{t, i}^{\tilde{h}}[\cdot]$ on both sides, and using the final condition (\ref{eq:finalCondition}), the above equation simplifies to:
\begin{equation}
 \begin{split}
\tilde{u}(t, i) &=  \mathbb{E}_{t, i}^{\tilde{h}} \left[ e^{\theta\int_t^T g(s, X_s, \tilde{h}_s)\ds} \right]\\
 &-\mathbb{E}_{t, i}^{\tilde{h}} \left[\int_t^T e^{\theta\int_t^s g(w, X_w, \tilde{h}_w) \rd w}\left(\frac{\partial \tilde{u}}{\partial s}(t, X_s) + A_i(\tilde{h}, \tilde{u})\right)\ds\right],
\end{split}
\end{equation}
where the expectation of the increment of a martingale is zero. From the definition of the ODE (\ref{eq:ode}) the following condition holds:
\begin{equation}
\frac{\partial \tilde{u}}{\partial s}(t, X_s) + A_i(h, \tilde{u}) \geq 0,
\end{equation}
with equality for the optimal strategy $\tilde{h}$. Hence we have:
\begin{equation}
 \begin{split}
\tilde{u}(t, i) &\leq  \mathbb{E}_{t, i}^{h} \left[ e^{\theta\int_t^T g(s, X_s, h_s)\ds} \right].
\end{split}
\end{equation}
with equality for the optimal strategy $\tilde{h}$:
\begin{equation}
 \begin{split}
\tilde{u}(t, i) &=  \mathbb{E}_{t, i}^{\tilde{h}} \left[ e^{\theta\int_t^T g(s, X_s, \tilde{h}_s)\ds} \right],
\end{split}
\end{equation}
which is equation (\ref{eq:value_f_u}) as required. 
This also shows that $\tilde{h}\in \mathcal{M}$ is optimal in the set of admissible strategies $\mathcal{H}$.
\end{proof}
We now turn to solving of the HJB equation (\ref{eq:ode}), (\ref{eq:finalCondition}).
\begin{lemma}[Lemma 2 from \cite{davis1998}]\label{lemma:operatorDifferentiability}
Let the operator $\mathcal{A}$ be given by (\ref{eq:caligraphicA}) and let both $A$ and $\mathcal{A}$ be differentiable in $\vec{u}$. 
Then for any $\vec{u}$ in the domain:
\begin{equation}
\left.\frac{d \mathcal{A}(\vec{u})}{d \vec{u}}\right|_{\vec{u}=u_0} = \left.\frac{d A(h^*, \vec{u})}{d \vec{u}}\right|_{\vec{u}=u_0},
\end{equation}
where $h^*(u_0)$ is the optimum at $\vec{u}=u_0$. Moreover, the result holds even if $h^*(\cdot)$ is not differentiable.
\end{lemma}
\begin{proof}(Sketch)
Because $h^*$ is optimal for $\vec{u}=u_0$, we have for all $\vec{u}$:
\begin{equation}\label{eq:lemmaInequality}
 \mathcal{A}(\vec{u})\leq A(h^*, \vec{u}),
\end{equation}
with equality at $\vec{u}=u_0$. The lemma follows from the fact that if the derivatives of $A(h^*, \cdot)$ and $\mathcal{A}$
were different at $u_0$, then the inequality (\ref{eq:lemmaInequality}) would fail in any neighbourhood of $u_0$. 
Note that the proof doesn't require the differentiability of $h^*(\cdot)$.
\end{proof}
\begin{proposition}\label{prop:lipschitz}
The operator $\mathcal{A}(\vec{u})$ is globally Lipschitz continuous on $\mathcal{U}$.
\end{proposition}
\begin{proof}
The derivative of the operator $\mathcal{A}(\vec{u})$ may be calculated using Lemma \ref{lemma:operatorDifferentiability}:
\begin{equation}
\frac{d \mathcal{A}(\vec{u})}{d \vec{u}} = \frac{d A(h^*, \vec{u})}{d \vec{u}} = A(h^*),
\end{equation}
where $h^*$ is the optimal trading strategy, which depends on the argument $\vec{u}$. By Proposition \ref{prop:unique_h}
it is well defined and unique on $\mathcal{U}$. 
To show that $ \mathcal{A}(\vec{u})$ is Lipschitz continuous we show that the derivative $A(h^*)$ is uniformly bounded. The lower bound of $A(h)$ comes from the fact that every
element of the matrix $A(h)$, as a function of $h$, is continuous and bounded from below. 
To show the upper bound take any value $h_0\in\mathcal{J}_h=\prod_{i=1}^N\mathcal{J}(i)$ and take the constant function $h_0(u) = h_0$.
Using the optimality of $h^*$ we have:
\begin{equation}
\begin{split}
\sup_{\vec{u}\in \mathcal{U}} \mathcal{A}(\vec{u}) &= \sup_{\vec{u}\in \mathcal{U}} \left\{\inf_{h\in\mathcal{J}_h} \left\{ A(\vec{u},h)(t, i) \right\} \right\}\\
&\leq \inf_{h\in\mathcal{J}_h} \left\{\sup_{\vec{u}\in \mathcal{U}} \left\{ A(h)\vec{u} \right\} \right\}\\
&\leq \inf_{h\in\mathcal{J}_h} \left\{ |A(h)| u_{max} \mathbf{1} \right\}\\
&\leq |A(h_0)| u_{max}\mathbf{1},
\end{split}
\end{equation}
where the absolute value in $|A(h)|$ is interpreted componentwise and $\mathbf{1}$ is an $N$-vector with every element being unity. 
Hence the value of the operator $\mathcal{A}$ is uniformly bounded for all $\vec{u}$. Let $A_{-} = \sup\left\{x: x\leq 0\wedge A(h)_{ij} \forall  h\in\mathcal{J}_h; \forall i, j\in \mathcal{N}\right\}$
be the lower bound of all the elements of $A(h)$ or zero if positive, 
and let $A_+ = \max \left\{ |A(h_0)| u_{max}\mathbf{1}\right\}$ be the biggest element of the vector bounding the operator $\mathcal{A}(\vec{u})$ from above.
From Proposition \ref{prop:domain} and from the obvious relationship $\mathcal{A}(\vec{u}) = A(h^*) \vec{u}$, using a straightforward combinatorial argument,
we get that for any $i, j\in \mathcal{N}$:
\begin{equation}
A(h^*)_{ij} \leq \frac{A_+}{u_{min}} + (N-1)A_- u_{max},
\end{equation}
hence $A(h^*)$ is uniformly bounded from above and the proposition follows.
\end{proof}
The following result completes the proof of Theorem \ref{theorem:control}.
\begin{proposition}\label{prop:solution}
The differential equation (\ref{eq:ode}) with boundary condition (\ref{eq:finalCondition}) has a unique solution on $[0, T]$. 
\end{proposition}
\begin{proof}
First note that the final value $\vectFunc{u}(T) = \mathbf{1}$ is an interior point of $\mathcal{U}_0$, hence 
by Proposition \ref{prop:lipschitz} the ODE has a unique solution on some interval $[t, T]$. Let $\tau\in(0,T)$ be  the minimal time that this holds.
By Proposition \ref{prop:verification}, $\vectFunc{u}$ is the value function of the optimal investment problem on $[\tau, T]$.

Because the ODE has a unique solution at $\tau$ and by Proposition \ref{prop:unique_h}, the operator $\mathcal{A}$ is well defined at this point and 
the derivative $\frac{\partial \vectFunc{u}}{\partial t}$ is well defined at $\tau$, and hence $\vectFunc{u}(\tau)$ is an interior point of $\mathcal{U}_0$.
Therefore there exists some $s<\tau$, such that we may extend the solution to $[s, T]$, which contradicts the assumption that
$\tau$ was minimal. This finishes the proof that the ODE has a solution on the whole domain $[0, T]$.


\end{proof}
It is not generally possible to provide a closed form solution for the problem.
The value function is however given by an ordinary differential equation, which may be efficiently 
solved numerically. The special case with no jumps in asset prices has an easy solution, which we leave it to the reader to verify.
\begin{proposition}
In the special case with no jumps in asset prices, i.e. when $f(\cdot; i, j) \equiv 0$ for all $i, j \in \mathcal{N}$,
the HJB equation (\ref{eq:ode_min}) with final condition (\ref{eq:boundaryCondition}) defined above has a closed-form solution:
 \begin{equation}
\vectFunc{u}(t) = e^{(Q-\theta \diag (g^*))(T-t)},
\end{equation}
where
\begin{equation}
 g^*(i) = g(t, i, h^*_t(i)) = -\frac{1}{2(\theta+1)}(\mu-r\mathbf{1})'(\Sigma\Sigma')^{-1}(\mu-r\mathbf{1}) - r
\end{equation}
\end{proposition}
\noindent The original value function can be recovered using (\ref{eq:value_f_v}) and is given by:
\begin{equation}
\vectFunc{v}(t) = -\frac{1}{\theta} \log \vectFunc{u}(t) = -\frac{1}{\theta} (Q-\theta \diag (g^*))(T-t)
\end{equation}

\subsection{Case with independent jumps}
Let us now look at the well studied case, where the jumps in the assets do not coincide with the regime switches.
In every state the asset processes follow a jump diffusion.
The martingales $M_t$ defined in (\ref{eq:Mt}) become:
\begin{equation}
M_t = \sum_{T_{i}<t} Z_i - \int_0^t \lambda(X_{s-}) \int_{\mathbb{R}^m} z \gamma(z; X_{s-}) \rd z \ds,
\end{equation}
where $\lambda(X_{s-})$ is the jump intensity and $\gamma(z, X_{s-})$ is the jump distribution in state $X_{s-}$.
As a consequence the the function $g$ defined in (\ref{eq:g}) becomes:
\begin{equation}
\begin{split}
 g(t, i, h) &= \half (\theta +1)h'\Sigma(t,i)\Sigma(t,i)'h - r(t, i) - h'(\mu(t,i)-r(t,i)\mathbf{1})\\
 &+\frac{\lambda(i)}{\theta}\int_{\mathcal{Z}^i}[(1+h'z)^{-\theta} -1]\gamma(z; i) \rd z + \lambda(i) h'\xi(i, t). 
\end{split}
\end{equation}
Note also, that in this case the measure change $\chi^T_h$, defined in (\ref{eq:chi_t}), does not 
change the distribution, and hence the generator, of the factor process $X$. 
Operator $A$, originally defined in (\ref{eq:A}), becomes:
\begin{equation}
\begin{split}
A(\vec{u},h)(t, i) &=Q(i) \vec{u} +\theta g(t, h, i) \vec{u}(i)\\
&= \sum_{j\neq i}Q(i, j) [\vec{u}(j) - \vec{u}(i)]  \\
&+\half \vec{u}(i) \theta (\theta +1)h'\Sigma(t,i)\Sigma(t,i)'h \\
&-\theta \vec{u}(i)h'[\mu + \lambda(i)\xi(i, t) - r\mathbf{1}] \\
& -\theta \vec{u}(i)r(t, i) +\vec{u}(i)\lambda(i)\int_{\mathcal{Z}^i}[(1+h'z)^{-\theta} -1]\gamma(z; i) \rd z,
\end{split}
\end{equation}
To compare the risk profile of this classical case with the new model proposed in the main part of the paper 
we match the jump intensity and distribution in each regime in these two models.Hence we
set $\lambda(i) = \sum_{j\neq i}$ and $\lambda(i) \gamma(z; i) \equiv \sum_{j\neq i} Q(i, j) f(z; i, j)$.
The operator $A$ becomes:
\begin{equation}
\begin{split}
A(\vec{u},h)(t, i) &= \int_{\mathcal{Z}^i}[(1+h'z)^{-\theta}]\vec{u}(i)\sum_{j\neq i} Q(i, j) f(z; i, j) \rd z  \\
&+\half \vec{u}(i) \theta (\theta +1)h'\Sigma(t,i)\Sigma(t,i)'h \\
&-\theta \vec{u}(i)h'[\mu + \sum_{j\neq i} Q(i, j) \xi(i, t) - r\mathbf{1}] \\
& -\theta \vec{u}(i)r(t, i) + \sum_{j\neq i}Q(i, j) [\vec{u}(j) - 2\vec{u}(i)].
\end{split}
\end{equation}
First note that the choice of the optimal strategy is this case does not depend on the expected value of the criterion in other states
and hence the investor in this model is in this sense myopic.
Unlike the case with jumps coinciding with regime switches, this problem simplifies to finding optimal strategy in each regime independently.
Only after the regime switch the system moves to a different strategy. As we mentioned in the introduction, such a formulation of the model
is not realistic and doesn't help the investor hedge against regime changes.

Next, note that the equation above is very similar to (\ref{eq:A}), the main difference being the jump risk convex term. In the new model:
\begin{equation}
 \int_{\mathcal{Z}^i}[(1+h'z)^{-\theta}]\sum_{j\neq i} \vec{u}(j) Q(i, j) f(z; i, j) \rd z
\end{equation}
versus
\begin{equation}
 \int_{\mathcal{Z}^i}[(1+h'z)^{-\theta}]\sum_{j\neq i} \vec{u}(i) Q(i, j) f(z; i, j) \rd z
\end{equation}
in the case with independent jumps. Everything else being equal, the investor will be more risk averse in the model with coinciding jumps
if:
\begin{equation}
 \vec{u}(i) >  \vec{u}(j)\quad \textrm{for every }j\neq i
\end{equation}
that is if the value of $\vec{u}$ after the regime switch is higher than in the current state. Note that the unit of $u$ is opposite 
to utility, i.e. the higher $u$ the worse outcome. This will usually be the case in real applications, where regimes are used to 
model a potential market crash.

%
%

\section{Fixed-point and Mutual Fund characterisation of optimal strategies}\label{sect:mft}

\subsection{Kelly criterion}
The Kelly, or log-optimality, criterion corresponds to the maximisation of the logarithmic utility of wealth $\ex{\log V_T}$, and we have at \eqref{eq:logVt} an explicit expression for $\log V_T$, so the log-optimal strategy, a Markov strategy $h^K(t)=h^K(t, X_{t-})$, can be determined by pointwise maximization.

\begin{proposition}\label{prop:Kelly}
For each $(t,i)\in[0,T]\times\mathcal{N}$ the allocation $h=h^K(t,i)\in\mathbb{R}^m$ to the Kelly (log optimal) portfolio solves the fixed point problem, with $\Sigma=\Sigma(t,i)$,
\be\label{hK}
        h=(\Sigma\Sigma')^{-1}\left\{(\mu(t,i)-r(t,i)\mathbf{1})+\sum_{j\neq i} Q(i, j) \left[ \int_{\mathcal{Z}^i} \frac{z}{1+h'z}f(z; i, j) \rd z - \xi(i, j)\right]\right\}. \ee
\end{proposition}

\begin{proof}
The process $M_t$ is a martingale and hence for $h(t) \in \mathcal{H}$,
\begin{eqnarray}
\ex{\log V_t} 
&=& \mathbb{E}\Bigg[ \int_0^t[r_s+h'(\mu_s-r_s\mathbf{1})] \ds - \half\int_0^t h'_s\Sigma\Sigma'h_s \ds
                                                                                                        \nonumber\\
&& - \sum_{0\leq s\leq t} \left\{ \log (1+h_s'Z_s) - h_s'Z_s \right\} \Bigg]
                                                                                                        \nonumber\\
&=& \mathbb{E}\Bigg[ \int_0^t  \ell(s,X(s^-),h(s)) \ds \Bigg],
\end{eqnarray}
where
\begin{eqnarray}
        \ell(s,i,h)
        &:=& r(s,i)+h'(\mu(s,i)-r(s,i)\mathbf{1}) - \half h'\Sigma(s,i)\Sigma'(s,i)h 
                                                                                                        \nonumber\\
        &&      +\sum_{j\neq i} Q(i, j) \left[ \int_{\mathcal{Z}^i} \log(1+h'z)f(z; i, j) \rd z - h'\xi(i, j)\right] .   
                                                                                                        \nonumber
\end{eqnarray}
The optimal strategy can be calculated pointwise: with $\mathcal{J}(i)$ given by \eqref{eq:admissibility},
\begin{equation}
h^K(t,i) = \argmax_{h \in\mathcal{J}(i)} \ell(t, i, h).
\end{equation}
The functional $\ell$ is globally concave in $h$. As a result, it admits a unique maximiser. Applying the first order condition, we find that $\partial \ell(s,x,\hat{h})/\partial h=0$ if and only if $h$ satisfies the fixed-point relation \eqref{hK}.\end{proof}

Note that in the special case with no jumps in asset prices, i.e. when $f(\cdot; i, j) \equiv 0$ for all $i, j \in \mathcal{N}$, \eqref{hK} gives us the allocation of the Kelly portfolio explicitly as
\[  h^{KD}(t,i)= (\Sigma\Sigma')^{-1}\left(\mu(t,i)-r(t,i)\mathbf{1}\right).\]
This shows that when the investor directly observes the factor process---as we assume here---it is optimal to apply the Merton log-optimal strategy appropriate for parameters $\mu(t,i)$ etc. as long as $X_t=i$. There is no need to hedge against future changes in the factor process since, when such changes occur, there is no jump in asset prices  and the portfolio can be instantly rebalanced. When there are jumps in asset prices at the same time as a factor change, none of these statements apply and the investor must take precautionary measures.

\subsection{General case} Returning to the general case, the optimal strategy $h^*$ for the general case developed in Section~\ref{sec_general} is characterized by a fixed-point relation generalising \eqref{hK} of Proposition \ref{hK} above.
In the proof of Theorem~\ref{theorem:control} we saw that
\begin{equation}
 h^*(u, t, i) = \argmin_{h\in\mathcal{H}} \left\{ A(u,h)(t, i)\right\},
\end{equation}
where $A(u,h)(t, i)$ is defined at \eqref{eq:matrixA}, and we obtain the result below from the first-order conditions for this minimisation problem. The details are straightforward and are left to the reader. Similar calculations can be found in \cite{davis2013}, \cite{mdsl2014}.
\begin{proposition}\label{prop:fixedpoint}
For each $i\in\mathcal{N}$ the optimal asset allocation $h^*(t,i)$ solves the fixed point problem
\bea
 h &=& \frac{1}{1+\theta} (\Sigma\Sigma')^{-1} \Bigg[\mu(t,i) - r(t,i)\mathbf{1}\nonumber\\
&&  +\sum_{j\neq i}Q(i, j)\Bigg(\frac{\mathfrak{u}_j(t)}{\mathfrak{u}_i(t)} \int_{\mathcal{Z}^i}\frac{z}{(1+h'z)^{1+\theta}}f(z;i, j)\rd z-\xi(i, j)\Bigg)\Bigg],\label{**}
\eea
where $\mathfrak{u}(t)$ is the solution of the HJB equation \eqref{eq:ode}, \eqref{eq:finalCondition}. 
\end{proposition}

For $\theta>0$\, we can define an allocation $\tilde{h}^\theta:=((1+\theta)h^*-h^K)/\theta$, giving a representation for $h^*$ in the form
\[ h^*=\frac{1}{1+\theta}h^K+\frac{\theta}{1+\theta}\tilde{h}^\theta.\]
This is analogous to Merton's Mutual Fund Theorem \cite{mdsl2014} in that the investor keeps a constant fraction $1/(1+\theta)$ of his/her wealth in a Kelly fund and the complementary fraction in a hedge portfolio $\tilde{h}^\theta$. This is not a mutual fund theorem in Merton's sense: the hedge portfolio is not universal but depends on the investor's risk appetite (as the notation suggests). It does however demonstrate how increasing risk aversion causes the investor to move away from a pure optimal growth strategy in the direction of a strategy more concerned with hedging jump risk. Indeed, if there are no jumps in the asset prices then from \eqref{**} we have $h^*=h^K/(1+\theta)$ and $\tilde{h}^\theta=0$, so the `hedge portfolio' is cash and the investor is following Merton's strategy for power utility, with parameters appropriate to the current value of the factor process.

\section{Concluding remarks on application of the model}\label{sect:conc}
In the abstract of his 1952 paper \cite{mar52}, Harry Markowitz states
\begin{quote}
The process of selecting a portfolio may be divided into two stages. 
The first stage starts with observation and experience and ends with 
beliefs about the future performances of available securities. The 
second stage starts with the relevant beliefs about future performances 
and ends with the choice of portfolio.
\end{quote}
Obviously, this paper is a contribution to the second stage, but we might ask what `observation and experience' could lead to a model of the sort we propose. Our model is a highly stylized description of reality: no-one could believe that in reality there are a finite number of factors `out there' that control the market. For this reason our view is that models similar to ours, but with an unobserved factor process estimated by a Wonham-type nonlinear filter \citep{bc2009}, are misguided: the filter output could be meaningless when the filter input is not the process for which the filter is designed. Instead, our model is intended to scope out the future in a way that is good enough to enable the portfolio manager to take advantage of opportunities while also enabling him/her to hedge against excessive risk. There will be periods of high and low volatility, bull and bear markets and occasional market crashes. A good model should predict all of these with some realistic evaluation of their probabilities and timing. Our model seems general enough to accomplish this task while at the same time incurring a modest computational overhead, at least in comparison with standard jump-diffusion models where the HJB equation is a PIDE.

The same objectives are pursued from a somewhat different perspective in the literature on Stochastic Programming (see  \cite{wtz03} or  \cite{bl2011}). The models here are discrete state $j$ and discrete time $k$ with only a limited number of time stages, perhaps even 4 or 5. Often a model is constructed on the basis of a \emph{scenario tree}, a standard tree structure in which the stages represent real times $0,t_1,t_2,\ldots$ where the gaps $t_{k+1}-t_{k}$ increase with $k$ so that the model is more detailed for the near future and sketchy for the far future. At each stage $k+1$ market parameters are specified in each state $j$, to apply for the period $(t_{k+1}-t_{k}]$ and at time $k$ a random choice is made to determine which set of market parameters will apply. Each path through the tree is a \emph{scenario} and methods of mathematical programming are used to optimise some performance functional over the set of decision variables. Construction of scenario trees is a big subject in its own right, see for example  \cite{kw2007},  \cite{kb2001} or \cite{pf2001}.

Our approach takes a more probabilistic viewpoint in which the stages $k$ correspond to the random switching times $\{t: X_t\not= X_{t-}\}$ and within each stage the model undergoes a random evolution modelled in continuous time. It seems to us that, while stochastic programming is certainly the right approach for long-term asset and liability management, the probabilistic approach could be a superior alternative for shorter time horizons of a few years when the objective is simply wealth maximisation with no liability-induced constraints.

%
%

\appendix
\section{Measure change}\label{sect:measureChange}
In this appendix we summarise the measure change theory used in the paper. Section \ref{sect:measureChangeMC} is a special case of
the theory presented in \cite{davis2011b} and deals with finite state Markov chains. Section \ref{sect:measureChangeCP}
discusses the measure change induced by regime-switching compound Poisson process, in particular the impact on the underlying
Markov chain. Finally the last section discusses the risk-adjusted changes of measure.

\subsection{Measure change for Markov chains}\label{sect:measureChangeMC}
Let $X_t$ be a Markov chain with generator matrix $Q$, as described in Section \ref{sect:market}.
For any function $\xi:\mathcal{N}\times\mathcal{N}\times\mathbb{R}\rightarrow \mathbb{R}$ define a martingale $M^\xi$ as follows:
\begin{equation}\label{eq:Mf}
 M^\xi_t =  \sum_{T_{i} \leq t} \xi(X_{T_i-}, X_{T_i}, T_i) - \int_0^t \sum_{j\neq X_{s-}} \xi(X_{s-}, j, s)Q(X_{s-},j) \ds,
\end{equation} 
Next, define the measure change martingale as the stochastic exponential of $M^\xi$:
\begin{equation}\label{eq:dQdP}
\begin{split}
 \left.\frac{\rd \mathbb{Q}}{\rd \mathbb{P}}\right|_{\FT} &= \mathcal{E}(M^\xi)\\
&=e^{M_T^\xi - \half [M^\xi, M^\xi]^c_T} \prod_{0<s\leq T} (1+\Delta M_s^\xi)e^{-\Delta M^\xi_s}\\
&=\prod_{0<s\leq T} (1+\xi(X_{s-}, X_s, s))\exp \left(- \int_0^t \sum_{j\neq X_{s-}} \xi(X_{s-}, j, s)Q(X_{s-},j) \ds \right),
\end{split}
\end{equation}
where $\ex{\frac{\rd \mathbb{Q}}{\rd \mathbb{P}}}=1$, using arguments analogous to the proof of Proposition \ref{prop:martingale}. Note that by the properties of the generator of Markov chains
\begin{equation}
 \lambda(i) = -Q(i, i)
\end{equation}
is the intensity of jumps of the process $\{\Lambda\}$. Once the factor process jumps, the probability of jump from state $i$ to
another state $j$ is given by:
\begin{equation}
 P_{ij} = \begin{cases}
               \frac{Q(i, j)}{\lambda(i)}                & i \neq j\\
               0              & \text{otherwise}
           \end{cases}
\end{equation}
\begin{proposition}\label{prop:measureChangeMC}
In the new probability measure defined by (\ref{eq:dQdP}) the generator of $X$ is given by:
\begin{equation}
 \widetilde{q}_{ij} = \begin{cases}
               Q(i, j)(\xi(i, j, t)+1)                & i \neq j\\
               -\sum_{k\neq i}    \widetilde{q}_{ik}           & \text{otherwise}
           \end{cases}
\end{equation}
Let $\gamma$ and $\beta$ be such that:
\begin{equation}\label{eq:beta}
\beta(i) = \sum_j P_{ij} [\xi(i, j, t)+1]
\end{equation}
and
\begin{equation}\label{eq:gamma}
\gamma_{ij} = \frac{\xi(i, j, t)+1}{\beta(i)}
\end{equation}
Then, in particular, the jump intensity becomes $\widetilde{\lambda} = \beta \lambda$ and the probability of jumps become:
\begin{equation}
 \widetilde{P}_{ij} = \gamma_{ij}P_{ij},
\end{equation}
\end{proposition}
\begin{proof}
Note that (\ref{eq:beta}) and (\ref{eq:gamma}) imply that:
\begin{equation}
 \sum_j  \gamma_{ij}P_{ij} = 1.
\end{equation}
and
\begin{equation}
\xi(i, j, t) = \gamma_{ij}\beta(i)-1
\end{equation}
Hence the result is a special case of the change of measure theory presented in \cite{davis2011b}. 
\end{proof}
In our particular case the equation simplifies to:
\begin{equation}
 \left.\frac{\rd \mathbb{Q}}{\rd \mathbb{P}}\right|_{\FT} = \prod_{0<s\leq T} (1+\Delta M_s^\xi)e^{- \int_0^T \sum_j \xi(j, X_{s-}, s)P(X_{s-},j) \lambda(X_{s-}) \ds}
\end{equation}
Following \cite{cont2012} we can express the measure change in the exponential form:
\begin{equation}
\begin{split}
 \left.\frac{\rd \mathbb{Q}}{\rd \mathbb{P}}\right|_{\FT} &= \exp \left( \sum_{T_i<T} \log (1 + \Delta M^\xi_{T_i}) - \int_0^T \sum_j \xi(X_{s-}, j, s)P(X_{s-},j) \lambda(X_{s-}) \ds\right)\\
&=\exp \left( \sum_{T_i<T}\log (\gamma(X_{T_i-}, X_{T_i})\beta(X_{T_i-})) \right)\\
&\quad \times \exp \left(- \int_0^T \sum_j \xi(X_{s-}, j, s)P(X_{s-},j) \lambda(X_{s-}) \ds\right)
\end{split}
\end{equation}
The last integral might be simplified:
\begin{equation}
\begin{split}
 \sum_j \xi(X_{s-}, j, s)P(X_{s-},j) \lambda(X_{s-}) &= \lambda(X_{s-}) \sum_j (\gamma_{ij}\beta(i)-1) P(X_{s-},j)\\
&=\lambda(X_{s-}) \beta(i) \sum_j \gamma_{ij} P(X_{s-},j) \\
&\quad\quad - \lambda(X_{s-}) \sum_j P(X_{s-},j)\\
&=\widetilde{\lambda}(X_{s-}) - \lambda(X_{s-})
\end{split}
\end{equation}
After substitution and application of some simple algebra, we get the final form:
\begin{equation}
\begin{split}
 \left.\frac{\rd \mathbb{Q}}{\rd \mathbb{P}}\right|_{\FT}&= \exp \left(\sum_{T_i<T} \log \gamma(X_{T_i-}, X_{T_i})\right)\\
 &\quad\times  \exp \left(\sum_{T_i<T} \log \beta(X_{T_i-}) - \int_0^T \widetilde{\lambda}(X_{s-}) - \lambda(X_{s-})\ds\right).
 \end{split}
\end{equation}
Note, that the first term is responsible for the change of measure of the distribution of jump destination, and the
second term is just a Poisson process intensity change corresponding to the jump times.

\subsection{Measure change for regime-switching compound Poisson processes}\label{sect:measureChangeCP}
Now let the measure change be defined by:
\begin{equation}
\begin{split}
 \left.\frac{\rd \mathbb{Q}}{\rd \mathbb{P}}\right|_{\FT} &= \mathcal{E}(M_T)\\
&= e^{M_T - \half [M, M]^c_T} \prod_{0<s\leq T} (1+\Delta M_s)e^{-\Delta M_s},\\
\end{split}
\end{equation}
where $M_t$ is defined in (\ref{eq:Mt}) and $\mathbb{E}[\rd\mathbb{Q}/\rd\mathbb{P}]=1$, using arguments analogous to the proof of Proposition \ref{prop:martingale}.
If we denote by $\FT^X = \sigma(\{X_t\}_{0\leq t \leq T})$ the filtration generated by the factor process up to time $T$,
then the measure change relevant for process $\{X\}$ is given by:
\begin{equation}
\begin{split}
\left.\frac{\rd \mathbb{Q}}{\rd \mathbb{P}}\right|_{\FT^X} &= \ec{\mathcal{E}(M_T)}{\FT^X}\\
&=\exp(- \int_0^T  \sum_{j\neq X_{s-}}\xi(X_{s-}, j) Q(X_{s-}, j) \ds ) \ec{\prod_{0<s\leq T} (1+Z_s)}{\FT^X}\\
&=\exp(- \int_0^T  \sum_{j\neq X_{s-}}\xi(X_{s-}, j) Q(X_{s-}, j) \ds ) \prod_{0<s\leq T} \ec{(1+Z_s) }{\FT^X}\\
&=\prod_{0<s\leq T} (1+\xi(X_{s-}, j))\exp(- \int_0^T \sum_{j\neq X_{s-}}\xi(X_{s-}, j) Q(X_{s-}, j) \ds ),
\end{split}
\end{equation}
where we used the independence property of $Z_i$-s to interchange the product and expectation in the third line
and the definition of $\xi$ in (\ref{eq:xi}). 
Note that this measure change martingale is of the same form as used in the previous section:
\begin{equation}
 \left.\frac{\rd \mathbb{Q}}{\rd \mathbb{P}}\right|_{\Ft^X} = \ec{\mathcal{E}(M_T)}{\Ft^X} = \mathcal{E}(M^{\xi}_t),
\end{equation}
hence results from Proposition \ref{prop:measureChangeMC} apply.

\subsection{Risk-adjusted stochastic exponentials}
Given the factor process $X$ and the jump sequence $Z_i$ defined in Section \ref{sect:market}, let $M^{h, \theta}_t$ be a martingale given by:
\begin{equation}\label{eq:M_h_theta}
M^{h, \theta}_t = \sum_{T_{i}<t} [(1+h'Z_i)^{-\theta} -1] - \int_0^t \sum_{j\neq X_{s-}} \int_{\mathcal{Z}^i} [(1+h'Z_i)^{-\theta} -1] \phi(z; X_{s}) \rd z \ds,
\end{equation}
where 
\begin{equation}\label{eq:phi_z_i}
 \phi(z; i) = \lambda(i) \tilde{f}(z; i)
\end{equation}
is the compensator of jumps, $\lambda(i) = \sum_{j \neq i} Q(i, j)$ is the jump intensity in state $i$ and
\begin{equation}
\tilde{f}(z; i) = \frac{\sum_{j \neq i} Q(i, j) f(z; i, j)}{\lambda(i)} 
\end{equation}
is the (mixture) density of jump size in state $i$. The stochastic exponential of $M^{h, \theta}_t$ is given by:
\begin{equation}\label{eq:EMh_theta}
\mathcal{E}(M^{h, \theta}_t) =  \prod_{0<T_i\leq T} (1+h'Z_i)^{-\theta}e^{- \int_0^t \sum_{j\neq X_{s-}} \int_{\mathcal{Z}^i} [(1+h'Z_i)^{-\theta} -1] \phi(z; X_{s1}) \rd z \ds}.
\end{equation}
To see the impact of a measure change defined by this stochastic exponential on the factor process $X$, let:
\begin{equation}
\xi^\theta(X_{s-}, j) = \int_\mathbb{R} (1+h'z)^{-\theta} f(z; X_{s-}, j) \rd z -1.
\end{equation}
Using the results from the previous section, the projection to the filtration generated by the factor process is given by:
\begin{equation}
\ec{\mathcal{E}^\theta(M_t)}{F^X_T} = \prod_{0<T_i\leq T} (1+\xi^\theta(X_{T_i-}, X_{T_i}))e^{- \int_0^T \sum_{j\neq X_{s-}} \xi^\theta(X_{s-}, j) Q(X_{s-},j) \ds},
\end{equation}
because of the independence of the jump sizes $Z$ and the identity
\begin{equation}
\ec{(1+h'Z_i)^{-\theta}}{F^X_T} = \int_\mathbb{R} (1+h'z)^{-\theta} f(z; X_{T_i-}, X_{T_i}) \rd z = \xi^\theta(X_{T_i-}, X_{T_i}) +1
\end{equation}
for $T_i\leq T$. Hence the effect of this measure change on the factor process follows from Proposition \ref{prop:measureChangeMC}, where the martingale is defined by the function $\xi^\theta$.

The following proposition summarises the effect of this measure change on the distribution of jumps sizes $Z$; it is needed in the proof of Proposition \ref{prop:martingale}.
\begin{proposition}\label{prop:measure_change_phi}
Provided that $\mathcal{E}(M^{h, \theta}_t)$ is a martingale, let:
\begin{equation}
\frac{\rd \tilde{\mathbb{P}}}{\rd \mathbb{P}} = \mathcal{E}(M^{h, \theta}_t) 
\end{equation}
with $\mathcal{E}(M^{h, \theta}_t) $ defined in (\ref{eq:M_h_theta}). Then the compensator of jumps $Z$ in the new measure $\tilde{\mathbb{P}}$ is given by:
\begin{equation}
\tilde{\phi}(z; i)  = (1+h'z)^{-\theta}\phi(z; i)
\end{equation}
\end{proposition}
\begin{proof}
The change of measure formula for compound Poisson processes (extended with state-dependence) is given by:
\begin{equation}
\left.\frac{\rd \tilde{\mathbb{P}}}{\rd \mathbb{P}}\right|_{\Ft} = \prod_{s\leq t} \frac{\tilde{\phi}(Z_s; X_{s-})}{\phi(Z_s; X_{s-})} e^{(\lambda(X_{s-})-\tilde{\lambda}(X_{s-}))t},
\end{equation}
see e.g. \citet[pp. 498-499]{shreve2004}. Comparing to (\ref{eq:EMh_theta}) and using calculations similar to proof of Proposition \ref{prop:measureChangeMC}, we get that:
\begin{equation}
 \frac{\tilde{\phi}(z; X_{s-})}{\phi(z; X_{s-})} = (1+h'z)^{-\theta},
\end{equation}
which finishes the proof.

\end{proof}

\section{Proof of Proposition \ref{prop:martingale}}\label{sect:proof:prop:martingale}
This proof in an adaptation of \cite{klebaner2011} to the current setting.
First note that the process $\chi^h_t$, where $h\in \mathcal{H}$ is the fixed trading strategy, solves the following stochastic differential equation:
\begin{equation}
\rd \chi^h_t = \chi^h_{t-} \rd M^\chi_t,
\end{equation}
where
\begin{equation}
\begin{split}
M^\chi_t &= -\theta \int_0^t h_s'\Sigma\dWs + \sum_{0<s\leq T}[(1+h_s'Z_s)^{-\theta}-1] \\
&- \int_0^t \int_{\mathcal{Z}^i}[(1+h'z)^{-\theta} -1]\phi(z; X_{s-}) \rd z \ds
\end{split}
\end{equation}
where $\phi$ is defined as in (\ref{eq:phi_z_i}). First note that, using the assumption (\ref{eq:integrability}), the process $M^\chi$ is a martingale.
Let us define a localizing sequence of stopping times as:
\begin{equation}
 \tau_n = \inf \{t : \chi_t \geq n\}
\end{equation}
Then for every $n$, the stopped process $\chi_{(t\wedge \tau_n)-}$ is bounded by construction. 
The main idea behind the proof is that the uniform integrability of the family $\{\chi_{t\wedge \tau_n}\}_{n\rightarrow \infty}$ is verified by the de la Vall{\'e}e Poussin theorem with function $x\log x$  for $x>0$.
By Ito's lemma:
\begin{equation}
\begin{split}
\chi_{t\wedge \tau_n}^2 -1 &= -2\theta \int_0^t \ind{s\leq \tau_n}\chi^2_{(s\wedge \tau_n)-} h_s'\Sigma\dWs\\
&+2\sum_{0<s\leq T}\ind{s\leq \tau_n}\chi^2_{(s\wedge \tau_n)-}[(1+h_s'Z_s)^{-\theta}-1]\\
&-2\int_0^t\ind{s\leq \tau_n}\chi^2_{(s\wedge \tau_n)-} \int_{\mathcal{Z}^i}[(1+h'z)^{-\theta} -1]\phi(z; X_{s-}) \rd z \ds\\
&+\sum_{0<s\leq T}\ind{s\leq \tau_n}\chi^2_{(s\wedge \tau_n)-}[(1+h_s'Z_s)^{-\theta}-1]^2\\
&-\int_0^t\ind{s\leq \tau_n}\chi^2_{(s\wedge \tau_n)-} \int_{\mathcal{Z}^i}[(1+h'z)^{-\theta} -1]^2\phi(z; X_{s-}) \rd z \ds\\
&+\theta\int_0^t\ind{s\leq \tau_n}\chi^2_{(s\wedge \tau_n)-} h_s'\Sigma\Sigma'h_s \ds\\
&+\int_0^t\ind{s\leq \tau_n}\chi^2_{(s\wedge \tau_n)-} \int_{\mathcal{Z}^i}[(1+h'z)^{-\theta} -1]^2\phi(z; X_{s-}) \rd z \ds\\
\end{split}
\end{equation}
Note that the first five lines in the formula above form a martingale with expectation zero. Hence:
\begin{equation}
\begin{split}
 \ex{\chi_{t\wedge \tau_n}^2 -1} &= \mathbb{E}\left[\int_0^t\ind{s\leq \tau_n}\chi^2_{(s\wedge \tau_n)-} \Big( \theta h_s'\Sigma\Sigma'h_s\right.\\
 &\qquad\qquad\qquad\qquad\left. + \int_{\mathcal{Z}^i}[(1+h'z)^{-\theta} -1]^2\phi(z; X_{s-}) \rd z \Big) \ds\right].
\end{split}
\end{equation}
Using the assumptions from Section \ref{sect:market}, the following process is bounded for every state $i\in\mathcal{N}$ and time $s\leq \tau_n$:
\begin{equation}\label{eq:klebaner:bound}
\theta h_s'\Sigma(i)\Sigma(i)'h_s + \int_{\mathcal{Z}^i}[(1+h_s'z)^{-\theta} -1]^2\phi(z; i) \rd z \leq r,
\end{equation}
and so $\chi_{t\wedge \tau_n}$ is a square integrable martingale with $\ex{\chi_{t\wedge \tau_n}} = 1$. We can use it to define a measure change:
\begin{equation}
 \frac{\rd \mathbb{P}^n}{\rd \mathbb{P}} = \chi_{t\wedge \tau_n}
\end{equation}
From Girsanov Theorem the Brownian motion  in the new measure $\mathbb{P}^n$ for all $t\leq \tau_n$ is given by:
\begin{equation}\label{eq:W_n_t}
\tilde{W}_t = W_t + \theta \int_0^t h_s'\Sigma \ds
\end{equation}
and from Proposition \ref{prop:measure_change_phi}, the jump compensator becomes:
\begin{equation}\label{eq:tilde_phi}
\tilde{\phi}(z; i) = (1+h'z)^{-\theta}\phi(z; i) =\phi(z; i) + [(1+h'z)^{-\theta} -1 ]\phi(z; i)
\end{equation}
Note that $\chi_{t\wedge \tau_n}$ can be decomposed as:
\begin{equation}
\chi_{t\wedge \tau_n} = \exp(M_{t\wedge \tau_n} - A_{t\wedge \tau_n}),
\end{equation}
where
\begin{equation}
\begin{split}
M_{t\wedge \tau_n} &= -\theta \int_0^t h_s'\Sigma\dWs + \sum_{0<s\leq T}[(1+h_s'Z_s)^{-\theta}-1] \\
&- \int_0^t \int_{\mathcal{Z}^i}[(1+h'z)^{-\theta} -1]\phi(z; X_{s-}) \rd z \ds
\end{split}
\end{equation}
and
\begin{equation}
\begin{split}
A_{t\wedge \tau_n} &= \frac{\theta^2}{2} \int_0^t h_s'\Sigma \Sigma'h_s \ds\\
&+ \sum_{0<s\leq T}[(1+h_s'Z_s)^{-\theta}-1] -\log[(1+h_s'Z_s)^{-\theta}]. \\
\end{split}
\end{equation}
The elementary inequality $\log(x)\leq x-1$ for all $x>0$ implies that the process $A_t$ is non-negative. Therefore $\log(\chi_{t\wedge \tau_n})\leq M_{t\wedge \tau_n}$, and we have the bound:
\begin{equation}
 \ex{\chi_{t\wedge \tau_n} \log \chi_{t\wedge \tau_n}} \leq \ex{\chi_{t\wedge \tau_n} M_{t\wedge \tau_n}} = \mathbb{E}^n\left[ M_{t\wedge \tau_n}\right],
\end{equation}
where $\mathbb{E}^n\left[ \cdot \right]$ denotes the expectation in the $\mathbb{P}^n$ probability measure. Using (\ref{eq:W_n_t}) and (\ref{eq:tilde_phi}) 
we can write $M_{t\wedge \tau_n}$ as :
\begin{equation}
\begin{split}
M_{t\wedge \tau_n} &= -\theta \int_0^t h_s'\Sigma\rd \tilde{W}_s + \theta \int_0^t h_s'\Sigma\Sigma'h_s \ds + \sum_{0<s\leq T}[(1+h_s'Z_s)^{-\theta}-1] \\
&- \int_0^t \int_{\mathcal{Z}^i}[(1+h'z)^{-\theta} -1]\tilde{\phi}(z; X_{s-}) \rd z \ds\\
&+ \int_0^t  \int_{\mathcal{Z}^i}[(1+h'z)^{-\theta} -1 ]^2\phi(z; i)\rd z \ds,
\end{split}
\end{equation}
and so:
\begin{equation}
\begin{split}
 \mathbb{E}^n\left[ M_{t\wedge \tau_n}\right] &= \mathbb{E}^n\left[ \int_0^t \theta h_s'\Sigma\Sigma'h_s + \int_{\mathcal{Z}^i}[(1+h'z)^{-\theta} -1 ]^2\phi(z; i)\rd z \ds\right]\\
 &\leq r,
\end{split}
\end{equation}
by (\ref{eq:klebaner:bound}). This is a uniform bound, hence:
\begin{equation}
\sup_n \ex{\chi_{t\wedge \tau_n} \log \chi_{t\wedge \tau_n}} < \infty.
\end{equation}
As indicated above, this implies that the family $\{\chi_{t\wedge \tau_n}\}_{n\rightarrow \infty}$ is uniformly integrable by the de la Vall{\'e}e Poussin theorem,
which finishes the proof.

\end{document}